\documentclass[aps,prl,twocolumn,superscriptaddress,showkeys]{revtex4}
\usepackage{graphicx}
\usepackage{amsfonts}
\usepackage{amssymb}
\usepackage{makeidx}
\usepackage{amsmath}
\usepackage{hyperref}
\usepackage{verbatim}
\usepackage{ctable}
\usepackage{amsthm}
\newcommand{\beq}{\begin{equation}\ }
\newcommand{\eeq}{\end{equation}\ }

\newcounter{algorithm}
\begin{document}
\newtheorem{lemma}{Lemma}[section]
\newtheorem{define}{Definition}[section]
\newtheorem{cor}{Corollary}[section]
\newtheorem{claim}{Claim}
\newtheorem{theorem}{Theorem}[section]
\newtheorem{example}{Example}[section]

\title{A simple and numerical stable algorithm for solving the cone projection problem based on a Gram-Schmidt process}

\author{Demetris T. Christopoulos}
\affiliation{Department of Economics, National and Kapodistrian University of Athens}
\email[Demetris T. Christopoulos, ]{dchristop@econ.uoa.gr}
\date{\today}

\begin{abstract}\begin{center}Abstract\end{center}
We are presenting a simple and numerical stable algorithm for the solution of the cone projection problem which is suitable for relative small data sets and for simulation purposes needed for convexity tests. Not even one pseudo-inverse matrix is computed because of a proper Gram-Schmidt orthonormalization process that is used. 
\end{abstract}

\maketitle

\section{The Cone Projection problem\label{sec:geom}}

\hspace{0.5cm} We have the data set $(x_i,\phi_i),i=1,2,\ldots,n$ which has emerged from a convex function f at least $C^{(2)}[x_1,x_n]$ by the process:
\beq
\phi_{i}=f(x_{i})+\epsilon_{i},\,\,\epsilon\sim iid(0,\sigma^2\,I_n)
\eeq
We want to find the vector y that has the smallest euclidean distance from $\phi$ subject to the requirement of convexity $A\,y\geq{0}$, thus we have to solve the next primal optimization problem:\\
\beq
\begin{matrix}
\text{min}\,\left\{\sum_{i=1}^{n}\,\left(y_{i}-\phi_{i}\right)^{2}=\left(y-\phi\right)^{T}\,\left(y-\phi\right)\right\}\\
\\
\text{subject to:}\,\left(-A\,y\right)\leq{0} \\
\end{matrix}
\label{eq:primal1}
\eeq
There are two equivalent versions for the matrix A of the convexity inequalities constraints. The first one is is to observe that we have strict inequalities:
$$
x_{1}<x_{2}<\cdots<x_{n}
$$
so starting from the definition of convexity we proceed to the inequalities:\\
\beq
\label{eq:dxdx}
\begin{matrix}
{\frac {y_{{i+2}}-y_{{i+1}}}{x_{{i+2}}-x_{{i+1}}}}\geq {\frac {y_{{i+1}}-y_{{i}}}{x_{{i+1}}-x_{{i}}}} \\
\\
\left( y_{{i+2}}-y_{{i+1}} \right)  \left( x_{{i+1}}-x_{{i}} \right) \geq \left( y_{{i+1}}-y_{{i}} \right)  \left( x_{{i+2}}-x_{{i+1}} \right) \\
\\
\left( x_{{i+2}}-x_{{i+1}} \right) y_{{i}}+ \left( x_{{i}}-x_{{i+2}} \right) y_{{i+1}}+ \left( x_{{i+1}}-x_{{i}} \right) y_{{i+2}}\geq{0}\\
\end{matrix}
\eeq
By constructing now all the above inequalities for $i=1,2,\ldots,n-2$ we have formulated the matrix $A^{(i)}$.\\
\begin{widetext}
\beq
\label{eq:ai}
A^{(i)}=\begin{pmatrix}
x_{3}-x_{2}&x_{{1}}-x_{{3}}&x_{2}-x_{1}&0&\cdots&0 \\
0&x_{{4}}-x_{{3}}&x_{{2}}-x_{{4}}&x_{{3}}-x_{{2}}&0&0\\ 
\ddots&\ddots&\ddots&\ddots&\ddots&\ddots \\
0&0&0& x_{{n}}-x_{{n-1}} & x_{{n-2}}-x_{{n}} &x_{{n-1}}-x_{{n-2}} \\
\end {pmatrix}
\eeq 
\end{widetext}

The second way is obtained if we have equal spaced $x_i$-data. Then it is easy to eliminate the same positive quantity $\Delta{x}=x_{j+1}-x_{j}$ from all inequalities:
\beq
\label{eq:dxeq}
\begin{matrix}
\left( x_{{i+2}}-x_{{i+1}} \right) y_{{i}}+ \left( x_{{i}}-x_{{i+2}} \right) y_{{i+1}}+ \left( x_{{i+1}}-x_{{i}} \right) y_{{i+2}}\geq{0} \\
\\
\left( \Delta{x}\right)y_{{i}}-\left( 2\Delta{x}\right) y_{{i+1}}+\left( \Delta{x}\right)y_{{i+2}}\geq{0} \\
\\
y_{{i}}-2 y_{{i+1}}+ y_{{i+2}}\geq{0} \\
\end{matrix}
\eeq
again with $i=1,2,\ldots,n-2$ and create the matrix $A^{(ii)}$.\\
\beq
\label{eq:aii}
A^{(ii)}=\begin{pmatrix}
1&-2&1&0&\cdots&0 \\
0&1&-2&1&\cdots&0\\
\ddots&\ddots&\ddots&\ddots&\ddots&\ddots \\
0&0&\cdots&1&-2&1
\end {pmatrix}
\eeq 
 
\begin{lemma}\label{lem:lag}
The polar component $\rho^{*}$ of the vector decomposition $\phi=y^{*}+\rho^{*}$ with $y^{*}$ the solution of problem \ref{eq:primal1} is a linear combination of the negative rows of matrix A, while the coefficients are either zero (if the corresponding constraint is not binding - inactive) or positive (if the relevant constraint is binding - active).
\end{lemma}
\begin{proof}
The Lagrangian function of the problem and the first order condition for y can be written as:\\
\beq
L\left(y,\lambda\right)=\left(y-\phi\right)^{T}\,\left(y-\phi\right)+\lambda^{T}\,\left(-A\,y\right)
\label{eq:primal1L}
\eeq
\beq
\begin{matrix}
\frac{\partial{L\left(y,\lambda\right)}}{\partial{y}}=2\,\left(y-\phi\right)+\left(-A^{T}\right)\,\lambda=0 \\
\text{or}\\
2\,\left(y-\phi\right)+\left(-A^{T}\right)\,\lambda=0 \\
\text{or}\\
y=\phi+\frac{1}{2}\,A^{T}\,\lambda
\label{eq:primal1foc}
\end{matrix}
\eeq
So, for the optimal solution $\{y^{*},\lambda^{*}\}$ we have that it also holds:
\beq
\begin{matrix}
y^{*}=\phi+\frac{1}{2}\,A^{T}\,\lambda^{*} \\
\text{or}\\
\phi-y^{*}=-\frac{1}{2}\,A^{T}\,\lambda^{*} \\
\text{or}\\
\rho^{*}=\left(-A^{T}\right)\,\frac{\lambda^{*}}{2} \\
\text{or}\\
\rho^{*}=\left(-A^{T}\right)\,\hat{\lambda}^{*}
\end{matrix}
\label{eq:primal1opt}
\eeq
Thus the representation of the polar component of the data vector, following the definitions of  \cite{roc-70} and  \cite{mey-06} in the basis of the negative rows of A is half the Lagrange coefficient vector of the optimization problem \ref{eq:primal1}. The coefficients are zero or positive if the corresponding constraint is inactive or active respectively, due to Karush Kuhn Tucker complementarity slackness conditions.
\end{proof}

\section{A numerical stable Geometric Algorithm for Cone Projection}
\subsection{An illustrative example}
\begin{example}
Let' s start with a common convex function:
\beq
f(x)=x^{2}\, , x\in \left[0,1\right]
\eeq
\par \emph{\textit{Let the vectors $x=(0,\frac{1}{2},1,\frac{3}{2},2)$ and $\phi=(0,\frac{1}{2},\frac{5}{2},\frac{15}{4},4)$ as presented at Figure \ref{fig:conv5}}} where we have drawn also the chord connecting $(x_1,\phi_1)$ and $(x_5,\phi_5)$. If our data was convex then all $(x_i,\phi_i)$ should lie above the chord, so clearly we have not convexity here.\\
\begin{figure}[hbtp]
\begin{center}
\includegraphics[scale=0.4]{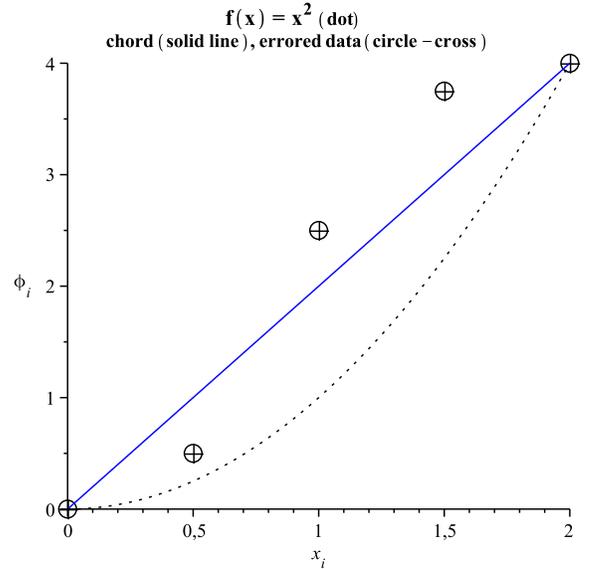}
\caption{The statement of the convex projection problem for $n=5$}\label{fig:conv5}
\end{center}
\end{figure}	 
The demand for convexity takes the form of the next inequality constraints, using matrix $A^{(ii)}$ because of the equal spaced $x_i$:
\beq
\begin{matrix}
A\,y\ge{0} \cr
\cr
A=\begin{pmatrix}
1&-2&1&0&0 \cr
0&1&-2&1&0 \cr
0&0&1&-2&1
\end{pmatrix}
\end{matrix}
\eeq
We define the matrix:
$$
R=-A=\begin{pmatrix}
-1&2&-1&0&0 \cr
0&-1&2&-1&0 \cr
0&0&-1&2&-1
\end{pmatrix}
$$
We pre-multiply vector $\phi$ by R:
$$
R\,\phi=\begin{pmatrix}-\frac{3}{2} \cr \frac{3}{4} \cr 1\end{pmatrix}
$$
If all components of the result were negative, then for the matrix A it should hold $A\,\phi\geq{0}$, so our data should be convex. So, if we seek for the greatest deviation from convexity then it is natural to pick the component that is the greatest positive. 	This is compatible with (i) deviation from convexity and with (ii) Lemma \ref{lem:lag}.\\
Here we observe the greatest entrance to be the 3rd one, so we pick up the 3rd row of $R$ as the best chance to obtain a component of the polar vector and proceed by taking the orthogonal projection of $\phi$ onto that row:
$$
\begin{matrix}
R_{1}=r_{3}^{T}=\begin{pmatrix}0 \\ 0 \\ -1 \\ 2 \\ -1 \end{pmatrix} \cr
\mu_{1}=\begin{pmatrix}\frac{\left\langle \phi,r_{3} \right\rangle}{\left\langle r_{3},r_{3} \right\rangle}\end{pmatrix}=\begin{pmatrix}\frac{1}{6}\end{pmatrix}=\begin{pmatrix}0.166666\ldots\end{pmatrix}=\begin{pmatrix}0.1\overline{6}\end{pmatrix}  \cr
\rho_{1}=R_{1}\,\mu_{1}=\begin{pmatrix}0 \\ 0 \\ -\frac{1}{6} \\  \frac{1}{3} \\  -\frac{1}{6} \end{pmatrix} \cr
y_{1}=\phi-\rho_{1}=\begin{pmatrix}0 \cr \frac{1}{2} \cr \frac{8}{3} \cr \frac{41}{12} \cr \frac {25}{6} \end{pmatrix} \cr
R_{1}\,y_{1}=\begin{pmatrix}-\frac{5}{3} \cr \frac {17}{12} \cr 0 \end{pmatrix}
\end{matrix}
$$
Now the greatest entry is the 2nd one, so we pick up the 2nd row of $R$ and continue by taking the matrix with the 2nd and 3rd rows of $R$. It is important to notice that we are always sorting our indices in ascending order.
$$
R_{2}=\begin{pmatrix}r_{2}^{T} & r_{3}^{T}\end{pmatrix}=\begin{pmatrix}0&0 \\ -1&0 \\ 2&-1 \\ -1&2 \\0&-1\end{pmatrix}
$$
Now there exist two ways for projecting our data $\phi$ on the two columns of $R_{2}$:\\
\begin{enumerate}
	\item The traditional way, i.e. the OLS estimator, which involves the pseudo inverse matrix and implies many numerical instabilities
	\item The new proposed way of taking the projection on the orthonormal base produced from them via the Gram-Schmidt procedure.
\end{enumerate}
We choose $2^{nd}$ way and first construct for $R_{2}$ with Gram-Schmidt the matrix with orthonormal columns:
$$
V=\begin{pmatrix}
0&0\\
-\frac{\sqrt {6}}{6}&-\frac{\sqrt {30}}{15}\\
\frac{\sqrt {6}}{3}&\frac{\sqrt {30}}{30}\\ 
-\frac{\sqrt {6}}{6}&\frac{2\,\sqrt {30}}{15}\\ 
0&-\frac{\sqrt {30}}{10}
\end{pmatrix}
$$ 
Then we just take the projections of $\phi$ on the two columns of V:\\
$$
\mu_{2}=V^{T}\,\phi=\begin{pmatrix}\frac{\sqrt {6}}{8} \cr \frac {3\,\sqrt {30}}{20} \end{pmatrix}=
\begin{pmatrix}0.3061862179 \\ 0.8215838362\end{pmatrix}
$$
The reader who is familiar with the active set methodology has to notice that our vector is not identical anymore to the $\lambda$ vector of that method, because of the orthonormalization process. This is the cost for the numerical stabilization of our algorithm.We continue executing our algorithm:
$$
\begin{matrix}
\rho_{2}=V\,\mu_{2}=\begin{pmatrix} 0 \\ -\frac {17}{40} \\ \frac{2}{5} \\ \frac {19}{40} \\ -\frac {9}{20} \end{pmatrix} \cr
y_{2}=\phi-\rho_{2}=\begin{pmatrix}\frac {37}{40} \\ \frac {21}{10} \\ \frac {131}{40} \\ \frac {89}{20} \end{pmatrix} \cr
R_{2}\,y_{2}=\begin{pmatrix}-\frac{1}{4} \\ 0 \\ 0 \end{pmatrix}
\end{matrix}
$$
We observe that there exist no polar edge vector to be inserted in our algorithm, so we exit with the solutions:
$$
\begin{matrix}
\rho^{*}=\rho_{2}=\begin{pmatrix} 0 \\ -\frac {17}{40} \\ \frac{2}{5} \\ \frac {19}{40} \\ -\frac {9}{20} \end{pmatrix}  \cr
y^{*}=y_{2}^{*}=\begin{pmatrix}\frac {37}{40} \\ \frac {21}{10} \\ \frac {131}{40} \\ \frac {89}{20} \end{pmatrix}
\end{matrix}
$$
We check again our solution:
$$
\begin{matrix}
A\,y^{*}=\begin{pmatrix}\frac{1}{4} \cr 0 \cr 0\end{pmatrix} \cr
\cr
\left\langle y^{*},\rho^{*} \right\rangle=0 
\end{matrix}
$$
It is expected to find that we will have two distinct lines for our cone projection plot and the second line has to be the OLS line for the set $\{(x_{i},\phi_{i}),i=2,3,4,5\}$, because only then we have a sequence of vanishing constraints. This fact is easily observed at Figure \ref{fig:gcx5}.
\begin{figure}[hbtp]
\begin{center}
\includegraphics[scale=0.4]{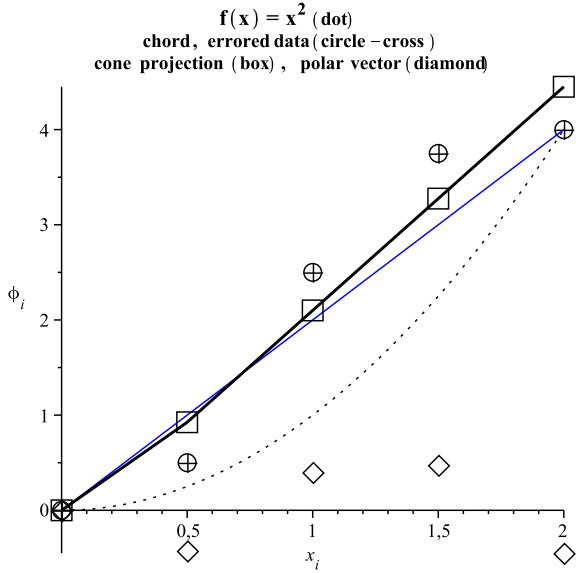}
\caption{The geometry of convex projection problem for $n=5$ in xy-plot}\label{fig:gcx5}
\end{center}
\end{figure}	
\end{example}

\subsection{The algorithm}
\hspace{0.25cm} By increasing the dimension of our problem until a rather big $n$ we continue to apply the same actions, i.e. we have established an \emph{algorithm} for cone projection.\\
\par We start by testing if our data is convex, so there is no need for cone projection at all. If it is not convex, then we multiply it by the $R$ matrix and seek for the maximum component and for its position. That direction is more probable to be an edge of the polar cone, so we find the projection of our data onto the $i^{th}$ row of $R$ matrix. Now we have found the first approximations of the vectors $\rho$ and $y=\phi-\rho$. We multiply again this $y$ with  $R$ ($b=R\,y$) and seek again for the maximum component and for its position. The new direction forms a set together with the previous one and we always sort the indices. The sorted indices construct the X matrix by taking the corresponding rows of $R$ matrix as the columns of X. Then we apply the Gram-Schmidt orthonormalization procedure on the columns of X and construct the matrix V. Now our $\mu$ vector can be calculated and then we find the next $\rho$ and $y$ approximation. We continue our algorithm until we reach at least one of the next three termination criteria:
\begin{enumerate}
	\item The algorithm is terminated if some $b$ vector is `practically' zero
	\item The algorithm is terminated if there is no improvement in the value of $b$
	\item The algorithm is terminated if next index i of $R-$row has already been chosen.
\end{enumerate}
\par Finally we exit from the algorithm with the set of indices J, where we have that the convexity constraints are satisfied as equalities (the active set indices, but without calculating the corresponding Lagrange coefficients), the polar component $\rho^{*}$ and the cone projection component $y^{*}$ of our initial data $\phi$. \emph{We do not compute even one time any kind of pseudo-inverse matrix, which is the fundamental tool of every regression technique.} This is due to the use of Gram-Schmidt orthonormalization process in order to do our orthogonal projections. This makes the algorithm \emph{numerical stable} for using it \emph{for simulation purposes}: we can establish the cone projection solution for every random set of vectors. This cannot be done with the traditional OLS solution, because of the existence of almost singular matrices for floating point arithmetic computations. The pseudo-code of the Algorithm is presented below.\\
\\
\label{alg:geomconv}
\refstepcounter{algorithm}
\begin{minipage}[htbp]{7cm}
\begin{center}
\emph{\textbf{A Gram-Schmidt polar basis\\ Cone Projection Algorithm}}\\ 
Find $y^{*}=\underbrace{argmin}_{y}\left\|y-\phi\right\|_{2}$\\
subject to $A\,y\ge{0}$\\
\end{center}
\textbf{INITIALIZE}\\
 $\left\{\epsilon_1,\epsilon_2, J=\{\},R=-A, b=R\,\phi,b_{old}=b+\theta,\theta>0\right\}$
\begin{itemize}
	\item \textbf{IF} $\left\{b\ge{0}\right\}$ \textbf{THEN} $\left\{\rho=0,y=\phi\right\}$ \textbf{BREAK}
	\item \textbf{ELSE}\\
	$$
	\begin{matrix}
	\textbf{FIRST PROJECTION}\\
	\begin{Bmatrix}
	s=\underbrace{max}_{j=1,\ldots,n-2}{b_{j}}&i=\underbrace{arg}_{j=1,\ldots,n-2}\left(b_{j}=s\right)&J=sort\left(J\cup\{i\}\right)
	\end{Bmatrix} \\
	\rho=project_{r_i}\phi \\
	y=\phi-\rho,b=R\,y\\
	\textbf{NEXT PROJECTIONS}\\
	\begin{Bmatrix}
	s=\underbrace{max}_{j=1,\ldots,n-2}{b_{j}}&i=\underbrace{arg}_{j=1,\ldots,n-2}\left(b_{j}=s\right)&J=sort\left(J\cup\{i\}\right)
	\end{Bmatrix} 
	\end{matrix}
	$$
	\textbf{IF} $\left\{s\le\epsilon_1\right\}$ \textbf{THEN BREAK}\\
  \textbf{DO WHILE} $\left\|b-b_{old}\right\|_{1}\ge\epsilon_2$ 
	$$
	\begin{matrix} 
	\begin{Bmatrix}
	X=\begin{pmatrix}r_{i_{1}}^{T} \ldots r_{i_{k}}^{T}\end{pmatrix}\,,\,i_{j}\in{J}\\
V=GramSchmidt(X)
\end{Bmatrix} \\
\rho=project_{V}\phi \\
y=\phi-\rho,b=R\,y\\
\begin{Bmatrix}
s=\underbrace{max}_{j=1,\ldots,n-2}{b_{j}}&i=\underbrace{arg}_{j=1,\ldots,n-2}\left(b_{j}=s\right)&J=sort\left(J\cup\{i\}\right)
\end{Bmatrix}
\end{matrix}
$$
\begin{itemize}
	\item \textbf{IF} $\left\{s\ge\epsilon_1\right\}$ \textbf{THEN}	$\left\{J=sort\left(J\cup\{i\}\right)\right\}$
	\item \textbf{IF} $\left\{i\in{J}\right\}$ \textbf{THEN BREAK}
	\item \textbf{ELSE BREAK}
\end{itemize}
\textbf{END DO}
\end{itemize}
\textbf{CHECK SOLUTION}\,\,$\left|\left\langle y,\rho\right\rangle\right|\leq{\epsilon_1}$\\
\textbf{RETURN} $\left\{J,\rho,y\right\}$\\
\end{minipage}
\\
We have developed our algorithm in four proper languages:
\begin{enumerate}
	\item First in Maple symbolic algebra system, where with just one page code we are able to execute our algorithm in absolute accuracy by using rational numbers as input data.
	\item Second in R suite, where we have floating point arithmetic, but it is an `alter ego' for the statistician community.
	\item Third in Matlab/Octave, for those who are familiar with the benefits of them. 
	\item Fourth in FORTRAN, one of the fastest ways to execute any numerical algorithm.
\end{enumerate}

\section{Conclusion}
The presented algorithm for solving the cone projection problem is quite simple because:
\begin{itemize}
	\item We don' t take care about the sign of Lagrange multipliers since we don' t compute them
	\item We just include one component of the polar basis every time
\end{itemize}
The algorithm is numerical stable for every kind of initial random vector $\phi$ because all projections are done via a Gram-Schmidt procedure and not with the common OLS pseudo-inverse matrix, which is very often close to singular for floating point arithmetic computations.\\
The algorithm is useful for convexity tests where we need to compute the weights of the weighted $\chi^2$ or Beta distribution that emerges for the corresponding statistical test, see for example \cite{mey-03}.

\end{document}